\documentclass[11pt]{article}
\usepackage{amsmath}
\usepackage{amsthm}
\usepackage{bcprules}
\usepackage{soul} %for underlining \ul{} and highlighting \hl{} 
\usepackage{color} %for \hl{} above

\newtheorem{theorem}{\textsc{Theorem}}[section]
\newtheorem{lemma}{\textsc{Lemma}}[section]

\title{Symmetric Encapsulated Multi-Methods\footnote{This work has been supported by a grant from Science Foundation Ireland.}  \footnote{This paper is a variant of the paper published in the OOPS track at SAC'2009 to provide full details of the type soundness proof presented in the original.}}

\author{David Lievens, Bill Harrison}
\date{}

\begin{document}

\maketitle

\begin{abstract}
In object systems, classes take the role of modules, and interfaces consist of methods. 
Because methods are encapsulated in objects, interfaces in object systems do not allow abstracting over \emph{where} methods are implemented. This implies that any change to the implementation structure may cause a rippling effect. Sometimes this unduly restricts the scope of software evolution, in particular for methods with multiple parameters where there is no clear owner. We propose a simple scheme where symmetric methods may be defined in the classes of any of their parameters. This allows client code to be oblivious of what class contains a method implementation, and therefore immune against it changing. When combined with multiple dynamic dispatch, this scheme allows for modular extensibility where a method defined in one class is overridden by a method defined in a class that is not its subtype. In this paper, we illustrate the scheme by extending a core calculus of class-based languages with these symmetric encapsulated multi-methods, and prove the result sound.
\end{abstract}

%%%%%%%%%%%%%%%%%%%%%%%%%%%%%%%%%%%%%%%%%%%%%%%%%%%%%%%%%%%%%%%%%%%%%%%%%%%%%%%
% | 
\newcommand{\sepbar}{\enskip | \enskip}
% e --> e'
\newcommand{\eval}[2]{\ensuremath{#1 \rightarrow #2}}
% Gamma |- t : C
\newcommand{\type}[3]{\ensuremath{#1 \vdash #2 : #3}}
% e -->* e'
\newcommand{\evalstar}[2]{\ensuremath{#1 \rightarrow^* #2}}

%%%%%%%%%%%%%%%%%%%%%%%%%%%%%%%%%%%%%%%%%%%%%%%%%%%%%%%%%%%%%%%%%%%%%%%%%%%%%%%

\section{Introduction}

% modular programming to make software evolution manageable
Evolution of large software can only stay manageable when the rippling effect of change is to some degree constrained. One way of achieving this is to program in terms of modules that communicate through interfaces~\cite{parnasmodules}. A change to an interface may still cause a rippling effect, but implementation details that are abstracted over (that are hidden behind the interface) may change without such a global impact. 

% modules in oo leak where methods are implemented
In object technology, classes act as modules, and interfaces are made of methods. However, because methods are encapsulated in objects, interfaces in object systems do not allow abstraction over \emph{where} a particular method is implemented. This implies that every time a change is made to the implementation structure of an application, i.e.~anytime a method is moved, a rippling effect may occur. 

Some movements of methods, however, do not fundamentally alter the implementation structure, and it may be desirable to make such changes without incurring a rippling effect. In particular, we are thinking of scenarios where methods are moved from the class of its receiving parameter to the class of any of its other parameters. As any practitioner can vouch, the choice of which class to assign responsibility for a method (especially one with multiple object parameters) may sometimes lie in the eye of the beholder. Hence it may be subject to change, either because of evolving insight into the design trade-offs, or simply for pragmatic reasons. 

A common way of side-stepping rippling change when a method implementation needs to move from one class to another is to write proxy-code to forward the method invocation to its new destination. Presence of proxy code may act as evidence that the flexibility of moving methods is desirable. Such proxy code frequently occurs in recently emerging application fields such as Ubiquitous Computing and Service-Oriented Computing. In fact, it was our experience in designing a programming language for these fields~\cite{continuum} that prompted our interest in this kind of malleability. Hence, one may assume that malleability of this kind will gain in importance together with further development of these application fields.

As an example, consider a method that registers a transaction between a \emph{shop}, an \emph{item} and a \emph{customer}. This method may be defined in the class that models shops. A client that invokes this method would therefore write something along the lines
\begin{verbatim}
  shop.transaction(item,customer)
\end{verbatim}
If it were later to be decided that this is not the optimal place for the method implementation to reside, and it was moved to the class representing items, then this client would break. Adding a proxy-method to the \emph{shop} class, although common, is not entirely desirable. The main issue lies with having to manually add and maintain the code, but more subtly, proxy code also essentially takes the form of double dispatch invocation and therefore has slightly different dispatch semantics from the original.

% we propose
%  use of symmetric methods
%  may be defined in class of any of the parameters
%  encapsulated in that class
% example: encapsulated in the item class
%   class-based access model may grant priviliged access to other parameters too if they are instances of the same class
To tackle this issue, we propose a scheme where methods are written in symmetric form and may be defined in the class of any of its parameters. With methods in symmetric form, we mean methods that list their receiving parameter explicitly, thereby treating all their parameters uniformly. Despite not giving any parameter special treatment, symmetric methods are encapsulated in the class in which they are defined. In other words, they have unrestricted access to the bowels of that class, and depend on the access model of the language for  access to the other parameters.

Under such a scheme, the above \emph{transaction} method may be defined in either the shop, item, or customer class. For example, it may be defined in the \emph{item} class as below, while not breaking a client that assumes the method being defined in the \emph{shop} class. 
\begin{verbatim}
  class Item {
    void transaction(Shop s, Item i, Customer c) {
      ... }
  }
\end{verbatim}

A main feature of object-orientation is its modular extensibility: method behaviour can be extended for derived classes (subtypes) without changing the original class. To achieve this, a subclass may define a method with the same\footnote{Or, depending on the flexibility, a \emph{conforming} signature.} signature that overloads the original method name with new behaviour. The discretely specified method bodies together make up the method and are called its branches. What method body is executed is decided at run-time based on the type of the receiving parameter\footnote{In addition, there may be a separate mechanism with which methods my be statically overloaded.}. 

We would like to generalize this ability to discretely specify method behaviour to our setting with encapsulated symmetric methods. For example, consider a \emph{charity shop} class that is a subclass of the \emph{shop} class. We would like to be able to define an extension of the \emph{transaction} method that deals specifically with charity shops in the \emph{charity shop} class, and have it invoked whenever the argument to the method is a charity shop. 
\begin{verbatim}
  class CharityShop extends Shop {
    void transaction(CharityShop cs, Item i, Customer c){
      ... }
  }
\end{verbatim}
To achieve this -while avoiding ad-hoc invocation semantics for different method call sites- we need to resort to multiple dynamic dispatch, also known as multi-methods. With multi-methods methods may be discretely specified, and branch selection is based on the dynamic type of all parameters. It is this characteristic that makes that branches may specialize any of the parameter types, and still be taken into account during selection. Note how the above method specializes a method that is not defined in its superclass, but in the unrelated \emph{item} class.

This brings us to a model with encapsulated symmetric multi-methods, in which the location of method implementations is partially abstracted over, and methods may be discretely specified along the hierarchies of the different argument classes. Clients, while still receiving the same level of feedback as before from static typing, can write code that is to a degree independent of where methods are implemented, and therefore may remain unchanged when the implementations are adapted or extended. The net effect of employing this scheme over standard object technology is that there is an increase in malleability of software, with the well-known software engineering benefits as a result. 

% outline of the paper
In the next sections, we incorporate the ideas above into a core object-oriented calculus (Featherweight Java). We start by discussing the merits of doing so, and then present the adapted formal semantics (Sect.~\ref{sec:syntax}-\ref{sec:typing}), including new notions of method lookup and selection (Sect.~\ref{sec:lookup}), and constraints on branches of overloaded methods (Sect.~\ref{sec:typing}). Soundness of the extension is tackled in Sect.~\ref{sec:soundness}. We discuss some observations about dispatching strategies in Sect.~\ref{sec:observations} Related work follows in Sect.~\ref{sec:related}, and finally we present our conclusions in Sect.~\ref{sec:conclusions}. The appendices contain an overview of the different formal definitions.

\section{Symmetric Featherweight Multi-Java}
\label{sec:sfmj}

We illustrate our ideas by extending Featherweight Java (FJ)~\cite{fjtoplas01} with symmetric encapsulated multi-methods, the result we call Symmetric Featherweight Multi-Java (SFMJ). In particular, $(i)$ we change the syntax for method definition and invocation to make all parameters explicit, $(ii)$ we allow methods to be defined in the class of any of its parameters, and $(iii)$ we introduce multiple dispatch.

We use FJ as a paradigmatic core calculus for mainstream class-based object-oriented languages. Our results are more generally applicable than to Java only. However, FJ proved to be useful vehicle for our explorations, because despite its simplicity it models method invocation in enough detail that we can adapt it for our purposes with little increase in complexity. In fact, the changes to the calculus are localized to method definition, invocation, and selection. As a result we may even establish soundness of SFMJ by a proof that is a limited extension to the one of FJ.

An in-depth discussion of FJ is beyond the merit of this article. We refer the reader to~\cite{fjtoplas01}, from which we also borrow the syntactic conventions without further comment. The structure of our exposition also mirrors the original one. An overview of the entire specification of SFMJ can be found in Fig.~\ref{fig:reduction} and Fig.~\ref{fig:typing}. To facilitate comparison with FJ, we also list the parts of the specification that have been altered in Fig.~\ref{fig:fj}.

\subsection{Syntax}
\label{sec:syntax}

Like FJ, an SFMJ program is a tuple $(CT,e)$ of a \emph{class table} and an expression. A class table is a mapping from class names to class declarations. The class table is assumed to be fixed, and implicitly accessible to all the semantic rules. Class declarations list the ancestor of the class (single inheritance), and in the body the fields, constructor and methods. There can only be a single constructor, and it takes a stylized form\footnote{See the FJ exposition for more details}. Methods may have any number of parameters (note the bar notation to indicate sequences), all with explicitly declared type, and a single return type. Method body consists of a single expression, which denotes the value returned on method invocation. The expressions consist of variables (which may be used to denote fields or arguments), field access, method invocation, object creation and casting. There is only one kind of values (expressions in normal form): objects, which are object creation expressions with their arguments fully evaluated. Note how objects carry their type, making it available at run-time. The syntax of SFMJ can be found in Fig.~\ref{fig:syntax} below. % move syntax into figure? 

\begin{figure}[h]
\fbox{
\begin{minipage}{0.95\linewidth}

\begin{align*}
L \quad ::= \quad
  & \text{class }C\text{ extends }C\,\{ \bar{C} \bar{f}; K \bar{M} \} & \text{class declarations} \\
K \quad ::= \quad
  & C(\bar{C} \  \bar{f})\,\{\text{super}(\bar{f}); \text{this}.\bar{f}=\bar{f};\} & \text{constructor declarations} \\
M \quad ::= \quad
  & C \  m(\bar{C} \ \bar{x}) \{\text{return }e;\} & \text{method declarations} \\
e \quad ::= \quad
  & x     & \text{variables} \\
  & e.f   & \text{field access} \\
%  & t.m(\bar{t}) & method invocation \\
  & m(\bar{e}) & \text{\hl{method invocation}} \\
  & \text{new }C(\bar{e}) & \text{object creation} \\
  & (C) e & \text{cast} \\
v \quad ::= \quad
  & \text{new }C(\bar{v}) & \text{objects}
\end{align*}

\end{minipage}
}
\caption{Syntax of SFMJ}
\label{fig:syntax}
\end{figure}

%\medskip

Note that the sole difference in syntax between FJ and SFMJ lies in method invocation: the original asymmetric syntax $t.m(\bar{t})$ has been changed to a function-like $m(\bar{t})$ that treats all parameters uniformly. Syntax for method definition, and the fact that methods may now be defined in more places can already be accommodated by existing syntax. 

% example program
% refer to example in introduction and note that it is written in this syntax

\subsection{Reduction}
\label{sec:reduction}

% reduction to values (v in the above)

The operational semantics of SFMJ is specified as a rewriting semantics. Reduction reduces expressions to values by successive applications of reduction rules. There is just one kind of values (object), as primitive values are not modelled. The full specification of the reduction relation can be found in Fig.~\ref{fig:reduction}. Only reduction of method invocation expressions are of interest here. The uniform treatment of parameters yields a single congruence rule \textsc{RC-Invk} (compare with \textsc{RC-Invk-Recv} and \textsc{RC-Invk-Arg} in Fig.~\ref{fig:fj}). The computation rule for method invocation expressions states simply that for fully evaluated method arguments\footnote{Unlike the original, our rule imposes call-by-value evaluation order, although this is immaterial for this story.}, the expression reduces to the substitution of these arguments into the body of the method retrieved by method lookup. 

Method lookup is specified in terms  the composition of a \emph{lookup} function that locates potentially applicable methods, and a \emph{selection} function that chooses one of these. To support multiple dispatch, the lookup function is dependent on the (run-time) types of the actual parameters\footnote{Note how values ($\text{new }C(\bar{v})$) inherently carry their run-time type.}. Definitions of these functions can be found in Sect.~\ref{sec:lookup} below.

\infrule[R-Invk]
 {v_i = \text{new }C(\bar{u}) \\
  \textit{select}(\textit{lookup}(m,\bar{C})) = B_0 \ m(\bar{B} \ \bar{x})\,\{\text{return }e_0;\}
 } 
 {\eval{m(\bar{v})}{[\bar{v}/\bar{x}]e_0}}

Note how the rule retains the basic structure of the original. In fact, with appropriate definitions for \emph{lookup} and \emph{select}, we may achieve the original semantics. Such a definition can be found in Fig.~\ref{fig:fjlookup}.

\subsection{Method Lookup and Selection}
\label{sec:lookup}

\subsubsection*{Method Lookup}

The method lookup function yields the set of all methods that are applicable to the arguments provided. Method lookup in SFMJ must consider that definition of methods may reside in the classes of any of their arguments. Also, because of overloading, a method name alone does not identify applicable methods. The lookup strategy consists of two parts. The top-level predicate \textit{lookup} takes a method name and a sequence of classes and returns the sequence of methods returned from successive applications of a support predicate $\textit{lookup}_1$ to each of the classes in the sequence. The $\textit{lookup}_1$ recursively inspects a class and all its superclasses for methods that satisfy a given signature.

%\noindent{\textbf{Method lookup: $\boxed{\textit{lookup}(m,\bar{C}) = \overline{ B_0 \ m(\bar{B} \bar{x} )  \,\{\text{return }e_0;\} }}$ }}

\infax[]
  {\textit{lookup}(m,\bar{C}) = \textit{lookup}_1(m,\bar{C},C_1), \dots, \textit{lookup}_1(m,\bar{C},C_n)}

\infax[]
  {\textit{lookup}_1(m,\bar{C},\text{Object}) = \bullet}

\infrule[]
  {CT(C) = \text{class }C\text{ extends }D\,\{\dots\} \\
   \bar{C} <: \bar{B} \andalso B_0 \ m(\bar{B} \bar{x})\,\{\text{return }t_0;\}  \in \bar{M}
  }
  {\textit{lookup}_1(m,\bar{C}, C) = B_0 \ m(\bar{B} \bar{x})\,\{\text{return }t_0;\}, \textit{lookup}_1(m,\bar{C},D)}

\infrule[]
  {CT(C) = \text{class }C\text{ extends }D\,\{\dots\} \\
   \bar{C} <: \bar{B} \andalso B_0 \ m(\bar{B} \bar{x})\,\{\text{return }t_0;\}  \not\in \bar{M}
  }
  {\textit{lookup}_1(m,\bar{C},C) = \textit{lookup}_1(m,\bar{C},D)}

\subsubsection*{Method Selection}

Method selection chooses from the set of applicable methods the \emph{most appropriate} one. This is the method with the signature that most closely matches the types of the arguments provided. In other words, it is the most specific method that is applicable to the given arguments. The predicate is defined as follows:

\infax[]
  {\textit{select}(\overline{ B_0 \ m(\bar{B} \ \bar{x} )  \,\{\text{return }e;\} }) = 
    B_{0,i} \ m(\bar{B}_i \ \bar{x}_i )  \,\{\text{return }e_i;\} 
    \\ \textit{s.t. } \bar{B}_i <: \bar{B}_j \ \forall j \in [1,n] }

The existence of a single most applicable method does not come for free, it must be asserted by adequate typing restriction, as we will discuss in Section~\ref{sec:typing}.

\subsection{Typing}
\label{sec:typing}

The bulk of the changes lie in the typing. This is unsurprising, as we allow methods to be defined in the class of any of its parameters, which induces changes to class and method typing. Additionally, the introduction of multiple dispatch requires the introduction of (global) constraints to avoid ambiguity between method branches. 

\medskip

%\noindent{\textbf{Expression typing: $\boxed{\type{\Gamma}{t}{C}}$}}
%\smallskip
\subsubsection*{Expression Typing}

Method invocation expressions are well-typed when the argument expressions are well-typed, and a method can be found whose whose formal argument types are supertypes of the actual argument types. The rule is very similar to the original FJ rule, except for the uniform treatment of parameters and the reliance on the previously defined lookup function. Note that the types passed to the lookup are the static (declared) types of the arguments. The run-time types of the actual arguments will be subtypes of these types. Hence, this check establishes that there will be at least one method applicable for this method invocation, although more methods may be taken into consideration during evaluation.

\infrule[T-Invk]
 {\type{\Gamma}{\bar{e}}{\bar{C}} \\
   B_0 \ m(\bar{B} \ \bar{x})\,\{\text{return }e;\} \in \textit{lookup}(m,\bar{C})  \\
  \bar{C} <: \bar{B}
 }
 {\type{\Gamma}{m(\bar{e})}{C}}

%\infrule[T-Invk]
% {\type{\Gamma}{\bar{e}}{\bar{C}} \\
%  \textit{select}(\textit{lookup}(m,\bar{C})) = B_0 \ m(\bar{B} \ \bar{x})\,\{\text{return }e;\}  \\
%  \bar{C} <: \bar{B_0}
% }
% {\type{\Gamma}{m(\bar{e})}{C}}

\medskip

%\noindent{\textbf{Method typing: $\boxed{M \  OK\text{ in }C}$}}
%
%\smallskip

\subsubsection{Method Typing}

A method definition is well-typed when $(i)$ the method body is well-typed assuming the types declared for the formal parameters, $(ii)$ the method body expression is assigned a type that is a subtype of the declared return type, and $(iii)$ the method is defined in the class of one of the parameters. It differs from the typing judgment in FJ, in its uniform treatment of the parameters (a pseudo-variable \texttt{this} need not adding to the typing context in which $t_0$ is typed), and in the absence of the invariance restriction on overloading that was imposed by the original. Because of support for multiple dispatch we have to move the equivalent of such a restriction to the program-level (see below).

\infrule[T-Meth]
 {\type{\bar{x}:\bar{C}}{e_0}{E_0} \andalso E_0 <: C_0 \\ 
  \exists i \in [1,n]. C \equiv C_i \\ %\text{ for \emph{an} }i \in [1,n] \\
 }
 {C_0 \  m(\bar{C} \ \bar{x}) \{\text{return }e_0;\} \ OK\text{ in }C}

% unchanged wrt to original 
%\noindent{\textbf{Class typing: $\boxed{C \ OK}$}}
%
%\infrule[] % for classes only
% {\textit{fields}(D) = \bar{D} \, \bar{g} \\
%  K = C(\bar{D}\,\bar{g}, \bar{C}\,\bar{f})\,\{\text{super}(\bar{g}); \text{this}.\bar{f}=\bar{f};\} \\
%  \bar{M} \, OK\text{ in }C
% }
% {\text{class }C\text{ extends }D\,\{\bar{C} \bar{f}; K \bar{M}\}\, OK}

%\noindent{\textbf{Program typing: $\boxed{(CT,e) \ OK}$}}
%
%\smallskip

\subsubsection{Program Typing}

In order to make a system with overloaded methods type safe and void of ambiguity problems, the set of branch definitions need to be subject to some restrictions. It must be noted that in our system, this set may not be obtained by looking only at the class in which a method is defined and the classes it references. This means that we loose separate compilation for SFMJ. To see this, consider the classes $A,B,A_1,B_1$ with $A_1 <: A$ and $B_1 <: B$. The following method definitions are legal:
\begin{verbatim}
  class A {
    void m(A a, B b) {...}
  }
  class B {...}
  class A1 extends A {
    void m(A1 a, B b) {...}
  }
  class B1 extends B {
    void m(A a, B1 b) {...}
  }
\end{verbatim}
The branch definitions in $A_1$ and $B_1$ give rise to an ambiguity when $m$ is invoked with instances of $A_1$ and $B_1$ as arguments. However, when inspecting the classes referenced by $A_1$, we do not encounter method branches defined in $B_1$ and vice versa. So, we need to look at all classes to make sure we do not miss any conflicting branch definitions. 

Note that the introduction of only multi-methods (and not symmetric methods) into FJ would not suffer from this, as all extensions of $m$ would be specified along one hierarchy, and therefore all (non-separate) branches would be visible from the leaf classes. The complication -but also the flexibility- of SFMJ derives from the fact that different hierarchies are involved, each of which may be independently extended.

We define the following support function to denote all the branches of an overloaded method, based on the observation that two method definitions are branches of the same overloaded method if there exist a method that they both specialize. Let $\mathcal{M}_{CT}$ be the set of all methods defined in a program $(CT,e)$, and $m$ be a member of that set. Then $\text{overloaded}(m)$ denotes the subset of $\mathcal{M}_{CT}$ that contains all the branches of the overloaded method of which $m$ is a member.
\begin{multline*}
\textit{overloaded}(C_0 \  m(\bar{C} \ \bar{x}) \{\text{return }e;\}) = \\ 
\{ B_0 \  m(\bar{B} \ \bar{x}) \{\text{return }e_1;\} \sepbar  
\exists D_0 \  m(\bar{D} \ \bar{x}) \{\text{return }e_2;\} \in M_{CT} \\
s.t. \bar{C} <: \bar{D} \text{ and }\bar{B} <: \bar{D}  \}
\end{multline*}

Castagna~\cite{castagna95calculus} has established the minimal constraints branches of an overloaded method must obey. Let $\mathcal O$ be a set of methods that are branches of an overloaded method. For every two methods $m_i$ and $m_j \in \mathcal O$, with respective argument types $\bar{I}, \bar{J}$, and return types $I_0, J_0$ it must hold that: 
\begin{enumerate}
	\item $\bar{K}$ maximal in $LB(\bar{I},\bar{J}) \Rightarrow$ a method $m$ with argument type $\bar{K}$ must be in $\mathcal O$
	\item $\bar{I} <: \bar{J} \Rightarrow I_0 <: J_0$ 
\end{enumerate}

The second constraint states that the value returned by a method that specializes another method should not violate the static return type of the method that is being specialized. The first constraint avoids ambiguity of method selection by guaranteeing that for any two branches there is always a branch that is most specific (be it possibly one of the two). 

We may elaborate a bit on the first constraint in our context to give an intuition on how this constraint makes that the lookup and selection functions defined in Section~\ref{sec:lookup} always yield a single method. Argument types are sequences of classes. The maximal lower bound of a sequence is the sequence of the maximal lower bounds, if they all exist. Because of single inheritance, a maximal lower bound of two classes may only exist if one class is a subclass of the other, and it is hence unique (if it exists). This means that $\bar{K}$ is unique, if it exists, and that it is a sequence of classes that are mentioned in $\bar{I}$ and $\bar{J}$. The lookup method defined above yields the subset of all branches that are applicable to the given arguments. The first constraint explicitly states that for any two methods that lookup returns, there exists a most specific one in the entire set of branches. Because this most specific branch has argument types that are mentioned in the argument types of the branches returned by lookup, and because methods are defined in the classes of their parameters, we know that lookup must also have returned this most specific branch. Hence \emph{lookup} and \emph{select} always yield a single method.

So, to make sure that the set of branches is subject to the above constraints, we introduce the following type judgment for programs. A set of branches \textit{wellformed} if they obey the above constraints.

\infrule[T-Prog]
 {\type{\emptyset}{e}{C} \\
  \forall m \in \mathcal{M}_{CT} . \textit{wellformed}(\textit{overloaded}(m))
 }
 {(CT,e) \ OK}

% explanation of why global knowledge is necessary: visibility of branches
% inferring of branches: both specialize common method
% well-formedness constraints (from castagna)
% how these constraints will make that the above lookup and selection always 'work'
% definition
% one predicate to denote set of all branches 
% another to enforce constraints
% constraints in most general form
% how this can be simplified in our framework

\subsection{Soundness}
\label{sec:soundness}

The soundness property that is established for FJ carries over entirely for SFMJ: if a term is well typed and it reduces to a normal form then it is either a value of a subtype of the original term's type, or an expression that gets stuck at a downcast. Formally:
%\medskip
%\textsc{Theorem} (SFMJ Type Soundness). If \type{\emptyset}{e}{C} and \evalstar{e}{e'} with $e'$ a normal form, then $e'$ is either a value $v$ with \type{\emptyset}{v}{D} and $D <: C$, or an expression containing $(D)\text{new }C(\bar{e})$ where $C <: D$.

\begin{theorem}[SFMJ Type Soundness]. If \type{\emptyset}{e}{C} and \evalstar{e}{e'} with $e'$ a normal form, then $e'$ is either a value $v$ with \type{\emptyset}{v}{D} and $D <: C$, or an expression containing $(D)\text{new }C(\bar{e})$ where $C <: D$.
\end{theorem}

\medskip

Having gone through the above exposition, the fact that this property carries over must not come as a surprise. We have kept the semantics of most constructs constant. The exception being method invocation, for which we provide similar guarantees as the original. 

The original proof (cf.~\cite{fjtoplas01} Sect.~2.4, and Appendix $A.1$) has two parts, first establishing \emph{Subject Reduction} and then \emph{Progress}. We can carry through the soundness proof for SFMJ using the original strategy. In fact, most of the proof needs no change. We will follow the original structure of the proof, comment on what can stay untouched, and only write out those bits that need changing.

%Neither proof poses significant challenges, and both carry through using the original strategy. Some of the support lemma's for Subject Reduction become slightly more complicated due to the fact that with overloaded methods both argument and return types may become smaller, rather than being invariant as in the original.

%Due to space restrictions we cannot discuss the proofs in more detail. We would like to refer to the accompanying technical report~\cite{sfmj-tr}\footnote{Reference removed for sake of blind refereeing.} for a more comprehensive treatment.

\begin{theorem}
(Subject Reduction). If \type{\Gamma}{e}{C} and \eval{e}{e'}, then \type{\Gamma}{e'}{C'} for some $C' <: C$.
\end{theorem}

\noindent The proof for this theorem relies on three support lemmas\footnote{In the original proof there are four support lemmas; we will not be needing the last one (Lemma A.1.4 in the paper), as it has to do with the typing of the implicit receiving parameter.}.

\medskip

\noindent\begin{lemma}[Specificity of Multi-method Dispatch]
If $\text{select}(\text{lookup}(m,\bar{C})) = B_0 \ m(\bar{B} \ \bar{x})\,\{\text{return }e_0;\}$, then for any $\bar{D} <: \bar{C}$ $\text{select}(\text{lookup}(m,\bar{D})) = E_0 \ m(\bar{E} \ \bar{x})\,\{\text{return }e_0';\}$ with $\bar{E} <: \bar{B}$ and $E_0 <: B_0$.
\end{lemma}

\noindent This lemma establishes that method lookup return a more specific method when given more specific argument types.

\begin{proof}
	
\noindent We proceed by induction on the \emph{number of positions} in which the input types ($\bar{C},\bar{D}$) differ.

subtype relation on input types: for the base case, we prove that the lemma holds for identical input types, while for the induction step, we prove the lemma holds for an input type that is specific in exactly one position.  then we `step up' from one input type, to an input type that is more specific in exactly one position.

\begin{itemize}
	\item The base case: the input types differ in $0$ positions, so $\bar{D} = \bar{C}$. \emph{select} returns the same method, so $\bar{B} = \bar{E}$.
	
	\item The induction step considers the case whereby the input types are in $k$ places more specific, assuming the lemma holds branches that are in $k-1$ places more specific. In particular, consider an input type $\bar{D'}$ that is in one place more specific than another input type $\bar{D}$: $\bar{D'} = D_0,\dots,D'_i,\dots,D_n$, with $D'_i <: D_i$.

	It is trivial to see that when invoked with $\bar{D'}$, the \emph{lookup} predicate returns a superset of the methods that it returns for $\bar{D}$: it returns (applicable) branches declared in $D'_i$ (and supertypes that are smaller than $D_i$), \emph{in addition} to the ones declared in $\bar{D}$, and supertypes.

	From this set, by definition, \emph{selects} retains the minimal elements. Given the well-formedness condition induced by \textsc{T-Prog}, we know that for all possible input types there is single most specific method branch, let us call it $E'_0\text{ }m\texttt{(}\bar{E'}\text{ }\bar{x}\texttt{)\{return }e_0''\texttt{;\}}$. We also know that because method branches are encapsulated, this most specific method branch is in the set returned by \emph{lookup}.

Hence we know that $\bar{E'} <: \bar{E}$, and because of induction hypothesis, $\bar{E} <: \bar{B}$.
	
\end{itemize}
	
\end{proof}

\medskip

\noindent\begin{lemma}[Term Substitution Preserves Typing] If \type{\Gamma,\bar{x}:\bar{B}}{e}{D}, and \type{\Gamma}{\bar{d}}{\bar{A}} where $\bar{A} <: \bar{B}$, then \type{\Gamma}{[\bar{d}/\bar{x}]e}{C} for some $C <: D$.
\end{lemma}

\noindent The lemma keeps its original form and proof method. The only case that changes is \textsc{T-Invk}, and only because of its reliance on Lemma 1 and the symmetric form of methods; the original reasoning for proving the case applies entirely. The proof  works by induction on the derivation of \type{\Gamma,\bar{x}:\bar{B}}{e}{D}.

\begin{proof} 
	
\begin{itemize}
	
\noindent We state only the case for \textsc{T-Invk}:
	
\item \emph{Case \textsc{T-Invk}:}
When the last step in the derivation is \textsc{T-Invk}, we have:
\infrule[T-Invk]
 {\type{\Gamma,\bar{x}:\bar{B}}{\bar{e}}{\bar{C}} \\
  \textit{select}(\textit{lookup}(m,\bar{C})) = B_0 \ m(\bar{B} \ \bar{x})\,\{\text{return }e_0;\} \\
  \bar{C} <: \bar{B}
 }
 {\type{\Gamma,\bar{x}:\bar{A}}{m(\bar{e})}{B_0}}

\noindent Because of induction hypothesis, we have 

\infax{\type{\Gamma}{[\bar{d}/\bar{x}]\bar{e}}{\bar{D}} \andalso \bar{D} <: \bar{C} }

\noindent From Lemma 1, we know that 

\infax{\text{select}(\text{lookup}(m,\bar{C})) = E_0 \ m(\bar{E} \ \bar{x})\,\{\text{return }e_0';\} }

\noindent with $\bar{E} <: \bar{B}$ and $E_0 <: B_0$. Using \textsc{T-Invk} on this, we may derive

\infax{\type{\Gamma}{m([\bar{d}/\bar{x}]\bar{e})}{E_0}}

\noindent and hence

\infax{\type{\Gamma}{[\bar{d}/\bar{x}]m(\bar{e})}{E_0}}

\end{itemize}

\end{proof}

\begin{lemma}[Weakening]. If \type{\Gamma}{e}{C}, then \type{\Gamma,x:D}{e}{C}. 
\end{lemma}

\begin{proof}
Original proof holds by relying on the modified lemmas. 
\end{proof}

\medskip

We can now give the proof for the Subject Reduction theorem. The proof proceeds by induction on the reduction derivation, with case analysis on the reduction rule used. The cases for \textsc{R-Field}, \textsc{R-Cast}, \textsc{RC-Field}, \textsc{RC-New-Arg}, and \textsc{RC-Cast} remain unchanged. The case for \textsc{RC-Invk} is a simple appeal to the induction hypothesis. This leaves only the case for \textsc{R-Invk}.

\begin{proof} 
	
\begin{itemize} \ \\

\item	\emph{Case \textsc{R-Invk}}: 
We assume that 
\begin{center}
	\begin{tabular}{ll}
		\eval{m(\overline{\text{new }C(\bar{u})})}{[\overline{\text{new }C(\bar{u})}/\bar{x}]e_0} & \\
		\type{\Gamma}{\overline{\text{new }C(\bar{u})}}{B_0} & (1)
	\end{tabular}
\end{center}
\noindent and set out to prove that 
$$\type{\Gamma}{[\overline{\text{new }C(\bar{u})}/\bar{x}]e_0}{A_0} \ \text{  for some }A_0 <: B_0$$

\noindent Assumption $(1)$ can only have been derived by application of \textsc{T-Invk}, and \textsc{T-New} on the first premise of that rule. So, we instantly get:
\begin{center}
	\begin{tabular}{cl}
		\type{\Gamma}{\overline{\text{new }C(\bar{u})}}{\bar{C}} & \\
		select(lookup($m,\bar{C}$)) = $B_0 \ m(\bar{B} \ \bar{x})\,\{\text{return }e_0;\}$ & \\
		$\bar{C} <: \bar{B}$ \quad & $(2)$
	\end{tabular}
\end{center}

\noindent Considering that we assume that methods are well-typed, we know from \textsc{T-Method} that $$\type{\bar{x}:\bar{B}}{e_0}{E_0} \text{with }E_0 <: B_0$$
\noindent Because of the Weakening lemma, we may derive
$$\type{\Gamma,\bar{x}:\bar{B}}{e_0}{E_0}$$
\noindent Using the Substitution lemma and $(2)$ , we obtain 
$$\type{\Gamma}{[\text{new }C(\bar{u})/\bar{x}]e_0}{A_0} \text{   with }A_0 <: E_0$$ 
\noindent Transitivity of $<:$ gives us the desired result. 

\end{itemize}

\end{proof}

\begin{theorem}
(Progress). Suppose $e$ is a well-typed expression.
\begin{enumerate}
 \item If $e$ includes $\text{new }C(\bar{e}).f$ as a subexpression, then fields($C$)=$\bar{C} \ \bar{f}$ and $f \in \bar{f}$ for some $\bar{C}$ and $\bar{f}$.
 \item If $e$ includes $m(\overline{\text{new }C(\bar{e})})$ as a subexpression, then select(lookup($m,\bar{C}$)) = $B_0 \ m(\bar{B} \ \bar{x})\,\{\text{return }e_0;\}$ , and $\#(\bar{x}) = \#(\bar{e})$ for some method $B_0 \ m(\bar{B} \ \bar{x})\,\{\text{return }e_0;\}$. %$\bar{x}$ and $e_0$.
\end{enumerate}
\end{theorem}

\noindent Only second case is altered.

\begin{proof}
The definition of \emph{lookup} guarantees that $\#(\bar{x}) = \#(\bar{e})$ for every method it returns. The constraints on branch definition guarantee that \emph{select} will return a single method.
\end{proof}

\medskip

The proof of type soundness follows immediately from the Subject Reduction and Progress theorems.

\section{Related Work}
\label{sec:related}

Generic methods in languages like CLOS~\cite{clos} provide similar abstraction over application structure as we propose here, and their software engineering benefits are well understood in the community. However generic methods are globally visible and not encapsulated, and therefore introduce state visibility issues (for a discussion see e.g.~\cite{binarymethods95}). 

The combination of encapsulation and the flexibility of changing application structure has been the aim of some research in the Aspect-Oriented community. In particular, there has been work on combining different class hierarchies without changing the classes being merged~\cite{sop,subjectod}. However, the language extensions and tools proposed by such research goes significantly beyond what we propose here, both in terms of flexibility and complexity.

% fortress
The Fortress programming language~\cite{fortressmod,fortress}, like SFMJ, uses symmetric encapsulated multi-methods. However, where we incorporate them in a class-only context, they propose a dedicated model based on components, traits and objects. Also, we provide a formal semantics and proof of soundness for our scheme.

% multiple dispatch in context of java (multijava, parasitic, nice prog lang, )
As noted before, symmetric methods may be employed with single dispatch. However, using multiple dispatch enriches the facility greatly. Multiple dispatch in the context of class-based programming languages has been extensively studied, both in terms of theoretical treatments~\cite{conflict}, and as practical extensions to mainstream languages to a.o.~C++~\cite{cpp}, Smalltalk~\cite{smalltalk}, and even ML (with added object-oriented features)~\cite{mlsub}. Two extensions to Java are of particular note.

Boyland et al. propose \emph{parasitic methods}~\cite{parasitic97}. Parasitic methods may dynamically overload other methods. Definition of parasitic methods is constrained to preserve modularity and separate compilation. The extension supports both covariant and contravariant specialization of methods. Some of the overloading strategies it supports, however, seem somewhat complex for programmers to confidently predict branch selection. 

% multijava as extension of java (including open methods)
MultiJava~\cite{multijavatoplas06} is an extension of Java to support symmetric multiple dispatch. Its emphasis lies on modularity and backward-compatibility. MultiJava also offers \emph{open classes}: classes to which methods can be added without changing the original class definition. MultiJava does not support symmetric methods or the definition of methods outside the class of the receiving argument. Open classes can only be extended by adding methods to the class; methods may not be moved to another class like our design allows.

% modularity and modular dispatch (dubious, relaxed multijava, jpred
There has been significant research focusing on the modularity issues associated with multi-methods. Millstein and Chambers~\cite{dubious} propose a number of different (language independent) constraints that yield different trade-offs between flexibility and modularity. Based on these ideas, they propose Relaxed MultiJava~\cite{relaxedmj03}. Finally, techniques from predicate dispatch -a generalization of type-based dispatch- have been brought to bear on the problem by Frost et al.~\cite{modularjpred}.

%%%%%%%%%%%%%%%%%%%%%%%%%%%%%%%%%%%%%%%%%%%%%%%%%%%%%%%%%%%%%%%%%%%%%%%%%%%%%%%
\section{Conclusions and Further Work}
\label{sec:conclusions}

In this article, we have argued that interfaces in mainstream object-oriented languages, by leaking \emph{where} methods are implemented, unduly restrict some useful forms of program evolution. We have proposed a simple scheme that mitigates this problem to some degree. The scheme consists of making it possible to define methods in the classes of any of its parameters. This means that a method may be moved to another of its parameter classes (but not any other classes) without breaking client code. To maintain object-oriented extensibility this implies the use of dynamic dispatch on all parameters. In the presence of multiple dispatch, a subclass of one of the parameter classes may define a method that overloads a method defined not in its superclass, but in another parameter class, and, for appropriate arguments, it would be this method that would be invoked.

To illustrate the idea as clearly as possible, and to address any doubts as to whether the above scheme is sound, we have defined our proposal as an extension of Featherweight Java. It turns out that the extension can be kept very simple, and that we may limit the changes to FJ to method definition, invocation, and lookup. Also, soundness may be established in a similar fashion to FJ and does not pose significant challenges.

%further work
Defining the scheme in terms of a core calculus has the disadvantage that only very basic object-oriented features are modelled. To come to a more interesting language, as further research, we would like to add more features. Some features may be added orthogonally. For example, there seems no indication that addition of imperative features is complicated by our scheme. Other features would be very easy to add, such as a null value. However, some features may pose a more significant challenge. For example adding support for abstract classes and interfaces may be difficult because of the presence of multiple dispatch. Abstract classes would introduce the \emph{coverage} problem for overloaded methods. In our scheme, we were always guaranteed to have an implementation of the most general branch of an overloaded method. With a naive integration with abstract classes this would not necessarily be always be the case. Interfaces would additionally incur the problem of multiple inheritance: spurious branches may have to be defined to avoid ambiguity problems. However, solutions to these problems have been proposed in different contexts, and we would like to investigate how these could be applied to our setting. In similar vein, while our scheme is modular, it does not allow for separate compilation of classes, as we need global knowledge of method definitions to make sure we do not miss conflicting method branch definitions. We would like to investigate how the research on modular multiple dispatch may be applied to lift this requirement.

%\cite{*}
% \bibliographystyle{splncs}
% \bibliography{oops09}

\begin{thebibliography}{10}

\bibitem{parnasmodules}
Parnas, D.L.:
\newblock On the criteria to be used in decomposing systems into modules.
\newblock Commun. ACM \textbf{15}(12) (1972)  1053--1058

\bibitem{fjtoplas01}
Igarashi, A., Pierce, B., Wadler, P.:
\newblock {Featherweight Java}: {A} minimal core calculus for {Java} and {GJ}.
\newblock TOPLAS \textbf{23}(3) (May 2001)  396--459

\bibitem{castagna95calculus}
Castagna, G., Ghelli, G., Longo, G.:
\newblock A calculus for overloaded functions with subtyping.
\newblock Information and Computation \textbf{117}(1) (15~February 1995)
  115--135

\bibitem{clos}
Bobrow, D.G., DeMichiel, L.G., Gabriel, R.P., Keene, S.E., Kiczales, G., Moon,
  D.A.:
\newblock Common lisp object system specification.
\newblock SIGPLAN Not. \textbf{23}(SI) (1988)  1--142

\bibitem{binarymethods95}
Bruce, K., Cardelli, L., Castagna, G., Leavens, G.T., Pierce, B.:
\newblock On binary methods.
\newblock Theor. Pract. Object Syst. \textbf{1}(3) (1995)  221--242

\bibitem{sop}
Harrison, W., Ossher, H.:
\newblock Subject-oriented programming: a critique of pure objects.
\newblock SIGPLAN Not. \textbf{28}(10) (1993)  411--428

\bibitem{subjectod}
Clarke, S., Harrison, W., Ossher, H., Tarr, P.:
\newblock Subject-oriented design: towards improved alignment of requirements,
  design, and code.
\newblock In: OOPSLA '99: Proceedings of the 14th ACM SIGPLAN conference on
  Object-oriented programming, systems, languages, and applications, New York,
  NY, USA, ACM (1999)  325--339

\bibitem{fortressmod}
Allen, E., Hallett, J.J., Luchangco, V., Ryu, S., Guy L.~Steele, J.:
\newblock Modular multiple dispatch with multiple inheritance.
\newblock In: SAC '07: Proceedings of the 2007 ACM symposium on Applied
  computing, New York, NY, USA, ACM (2007)  1117--1121

\bibitem{fortress}
Allen, Chase, Hallett, Luchangco, Maessen, Ryu, Steele, Tobin-Hochstadt:
\newblock The Fortress Language Specification. Version 1.0.
\newblock Sun Microsystems (2008)

\bibitem{conflict}
Castagna, G.:
\newblock Covariance and contravariance: conflict without a cause.
\newblock ACM Trans. Program. Lang. Syst. \textbf{17}(3) (1995)  431--447

\bibitem{cpp}
Pirkelbauer, P., Solodkyy, Y., Stroustrup, B.:
\newblock Open multi-methods for c++.
\newblock In: GPCE '07: Proceedings of the 6th international conference on
  Generative programming and component engineering, New York, NY, USA, ACM
  (2007)  123--134

\bibitem{smalltalk}
Foote, B., Johnson, R.E., Noble, J.:
\newblock Efficient multimethods in a single dispatch language.
\newblock In: ECOOP. (2005)  337--361

\bibitem{mlsub}
Bourdoncle, F., Merz, S.:
\newblock Type checking higher-order polymorphic multi-methods.
\newblock In: POPL '97: Proceedings of the 24th ACM SIGPLAN-SIGACT symposium on
  Principles of programming languages, New York, NY, USA, ACM (1997)  302--315

\bibitem{parasitic97}
Boyland, J., Castagna, G.:
\newblock Parasitic methods: an implementation of multi-methods for java.
\newblock In: OOPSLA '97: Proceedings of the 12th ACM SIGPLAN conference on
  Object-oriented programming, systems, languages, and applications, New York,
  NY, USA, ACM (1997)  66--76

\bibitem{multijavatoplas06}
Clifton, C., Millstein, T., Leavens, G.T., Chambers, C.:
\newblock Multijava: Design rationale, compiler implementation, and
  applications.
\newblock ACM Trans. Program. Lang. Syst. \textbf{28}(3) (2006)  517--575

\bibitem{dubious}
Millstein, T.D., Chambers, C.:
\newblock Modular statically typed multimethods.
\newblock In: ECOOP '99: Proceedings of the 13th European Conference on
  Object-Oriented Programming, London, UK, Springer-Verlag (1999)  279--303

\bibitem{relaxedmj03}
Millstein, T., Reay, M., Chambers, C.:
\newblock Relaxed multijava: balancing extensibility and modular typechecking.
\newblock In: OOPSLA '03: Proceedings of the 18th annual ACM SIGPLAN conference
  on Object-oriented programing, systems, languages, and applications, New
  York, NY, USA, ACM (2003)  224--240

\bibitem{modularjpred}
Frost, C., Millstein, T.:
\newblock Modularly typesafe interface dispatch in jpred.
\newblock In: The 2006 International Workshop on Foundations and Developments
  of Object-Oriented Languages (FOOL/WOOD '06). (January 2006)

\end{thebibliography}

\begin{figure}
\fbox{
\begin{minipage}{\linewidth}%{0.5\linewidth}

\noindent{\textbf{Computation Rules: $\boxed{\eval{e}{e'}}$}}

\infrule[R-Field]
 {\textit{fields}(C) = \bar{C} \, \bar{f}}
 {\eval{(\text{new }C(\bar{v})).f_i}{v_i}}

\infrule[R-Invk]
 {v_i = \text{new }C(\bar{u}) \\
  \textit{select}(\textit{lookup}(m,C)) = B_0 \ m(\bar{B} \ \bar{x})\,\{\text{return }e_0;\}
 } 
 {\eval{m(\bar{v})}{[\bar{v}/\bar{x}]e_0}}

\infrule[R-Cast]
 {C <: D}
 {\eval{(D)(\text{new }C(\bar{v}))}{\text{new }C(\bar{v})}}

\noindent{\textbf{Congruence Rules: $\boxed{\eval{e}{e'}}$}}

\infrule[RC-Field]
 {\eval{e_0}{e_0'}}
 {\eval{e_0.f}{e_0'.f}}

\infrule[RC-Invk]
 {\eval{e_i}{e_i'}}
 {\eval{m(\bar{v},e_i,\bar{e})}{m(\bar{v},e_i',\bar{e})}}

\infrule[RC-New-Arg]
 {\eval{e_i}{e_i'}}
 {\eval{\text{new }C(\bar{v},e_i,\bar{e})}{\text{new }C(\bar{v},e_i',\bar{e})}}

\infrule[RC-Cast]
 {\eval{e_0}{e_0'}}
 {\eval{(C)e_0}{(C)e_0'}}  

\end{minipage} 
}
\caption{SFMJ: Reduction rules}
\label{fig:reduction}
\end{figure}

\begin{figure}
\fbox{
\begin{minipage}{0.98\linewidth}

\noindent{\textbf{Method lookup: $\boxed{\textit{lookup}(m,\bar{C}) = \overline{ B_0 \ m(\bar{B} \bar{x} )  \,\{e_0;\} }}$ }}

\infax[]
  {\textit{lookup}(m,\bar{C}) = \textit{lookup}_1(m,\bar{C},C_1), \dots, \textit{lookup}_1(m,\bar{C},C_n)}

\infax[]
  {\textit{lookup}_1(m,\bar{C},\text{Object}) = \bullet}

\infrule[]
  {CT(C) = \text{class }C\text{ extends }D\,\{\dots\} \\
   \bar{C} <: \bar{B} \andalso B_0 \ m(\bar{B} \bar{x})\,\{\text{return }e_0;\}  \in \bar{M}
  }
  {\textit{lookup}_1(m,\bar{C}, C) = B_0 \ m(\bar{B} \bar{x})\,\{\text{return }e_0;\}, \textit{lookup}_1(m,\bar{C},D)}

\infrule[]
  {CT(C) = \text{class }C\text{ extends }D\,\{\dots\} \\
   \bar{C} <: \bar{B} \andalso B_0 \ m(\bar{B} \bar{x})\,\{\text{return }e_0;\}  \not\in \bar{M}
  }
  {\textit{lookup}_1(m,\bar{C},C) = \textit{lookup}_1(m,\bar{C},D)}
\end{minipage}
}
\caption{Lookup in SFMJ}
\label{fig:lookupsfmj}
\end{figure}

\begin{figure}
\fbox{
\begin{minipage}{0.98\linewidth}%{0.5\linewidth}
%\noindent{\textbf{Subtyping: $\boxed{C <: D}$}}
%
%\infax[]
% {C <: C}
%
%\infrule[]
% {C <: D \andalso D <: E}
% {C <: E}
%
%\infrule[]
% {CT(C) = \text{class }C\text{ extends }D\,\{\dots\}}
% {C <: D}
%
%\medskip

\noindent{\textbf{Expression typing: $\boxed{\type{\Gamma}{t}{C}}$}}

\infrule[T-Var]
 {x:C \in \Gamma}
 {\type{\Gamma}{x}{C}}

\infrule[T-Field]
 {\type{\Gamma}{e_0}{C_0} \andalso \textit{fields}(C_0) = \bar{C} \ \bar{f}}
 {\type{\Gamma}{e_0.f_i}{C_i}}

\infrule[T-Invk]
 {\type{\Gamma}{\bar{e}}{\bar{C}} \\
  B_0 \ m(\bar{B} \ \bar{x})\,\{\text{return }e;\} = \textit{select}(\textit{lookup}(m,\bar{C})) \\
  \bar{C} <: \bar{B_0}
 }
 {\type{\Gamma}{m(\bar{e})}{C}}

\infrule[T-New]
 {\textit{fields}(C) = \bar{D} \ \bar{f} \andalso
  \type{\Gamma}{\bar{e}}{\bar{C}} \andalso C <: D
 }
 {\type{\Gamma}{\text{new }C(\bar{t})}{C}}

\infrule[T-UCast]
 {\type{\Gamma}{e_0}{D} \andalso D <: C}
 {\type{\Gamma}{(C)e_0}{C}}

\infrule[T-DCast]
 {\type{\Gamma}{e_0}{D} \andalso C <: D \andalso C \neq D}
 {\type{\Gamma}{(C)e_0}{C}}

\infrule[T-SCast]
 {\type{\Gamma}{e_0}{D} \andalso C \not<: D \andalso D \not<: C \andalso
  \textit{stupid warning}
 }
 {\type{\Gamma}{(C)e_0}{C}}
 
\medskip

\noindent{\textbf{Method typing: $\boxed{M \  OK\text{ in }C}$}}

\medskip

\infrule[T-Meth]
 {\type{\bar{x}:\bar{C}}{e_0}{E_0} \andalso E_0 <: C_0 \\ 
  \exists i \in [1,n]. C \equiv C_i \\ %\text{ for \emph{an} }i \in [1,n] \\
 }
 {C_0 \  m(\bar{C} \ \bar{x}) \{\text{return }e_0;\} \ OK\text{ in }C}

\medskip

\noindent{\textbf{Class typing: $\boxed{C \ OK}$}}

\medskip

\infrule[]
 {\textit{fields}(D) = \bar{D} \ \bar{g} \\
  K = C(\bar{D}\,\bar{g}, \bar{C}\,\bar{f})\,\{\text{super}(\bar{g}); \text{this}.\bar{f}=\bar{f};\} \\
  \bar{M} \, OK\text{ in }C
 }
 {\text{class }C\text{ extends }D\,\{\bar{C} \bar{f}; K \bar{M}\}\, OK}

\medskip

\noindent{\textbf{Program typing: $\boxed{(CT,e) \ OK}$}}

\medskip

\infrule[T-Prog]
 {\type{\emptyset}{e}{C} \\
  \forall m \in \mathcal{M}_{CT} . \textit{wellformed}(\textit{overloaded}(m))
 }
 {(CT,e) \ OK}

\end{minipage}
} % end fbox
\caption{SFMJ: Typing rules}
\label{fig:typing}
\end{figure}% figure* for full-width figure

\begin{figure}
\fbox{
\begin{minipage}{0.98\linewidth}

% method invocation reduction rules
\noindent\textbf{Reduction Rules}

\infrule[RC-Invk-Recv]
 {\eval{e_0}{e_0'}}
 {\eval{e_0.m(\bar{e})}{e_0'.m(\bar{e})}}

\infrule[RC-Invk-Arg]
 {\eval{e_i}{e_i'}}
 {\eval{e_0.m(\bar{v},e_i,\bar{e})}{e_0.m(\bar{v},e_i',\bar{e})}}

\infrule[R-Invk]
 {v_i = \text{new }C_i(\bar{u}_i) \\
  \textit{select}(\textit{lookup}(m,C_0)) = B \ m(\bar{B} \ \bar{x})\,\{\text{return }e_0;\}
 }
 {\eval{(\text{new }C(\bar{u})).m(\bar{v})}
       {[ \text{new }C(\bar{u}) / \texttt{this},\bar{v}/\bar{x} ]e_0}}

\medskip

% method invocation typing
\noindent\textbf{Expression Typing}

\medskip

\infrule[T-Invk]
 {\type{\Gamma}{\bar{e_0}}{\bar{C_0}} \\
  B_0 \ m(\bar{B} \ \bar{x})\,\{\text{return }e;\} = \textit{select}(\textit{lookup}(m,C_0)) \\
  \type{\Gamma}{\bar{e}}{\bar{C}} \\
  \bar{C} <: \bar{B_0}
 }
 {\type{\Gamma}{e_0.m(\bar{e})}{C}}

% method typing
\noindent\textbf{Method Typing}

\infrule[T-Meth]
 {\type{\texttt{this}:C, \bar{x}:\bar{C}}{e_0}{E_0} \andalso E_0 <: C_0 \\ %body must be well-typed
  C \equiv C_i \text{ for \emph{an} }i \in [1,n] \\
  CT(C) = \text{class }C\text{ extends }D\,\{\dots\} \\ % when m defined in superclass, args must match
  \text{if }\text{select}(\text{lookup}(m,D) = D_0 \  m(\bar{D} \ \bar{x}) \{\text{return }e_0';\} \\
     \text{ then } \bar{C} = \bar{D} \text{ and } C_0 = D_0 
 }
 {C_0 \  m(\bar{C} \ \bar{x}) \{\text{return }e_0;\} \ OK\text{ in }C}

\end{minipage}
}
\caption{Aspects of FJ that Have Changed}
\label{fig:fj}
\end{figure}

\begin{figure}
\fbox{
\begin{minipage}{0.98\linewidth}

\noindent{\textbf{Method lookup: $\boxed{\textit{lookup}(m,\bar{C}) = \overline{ B_0 \ m(\bar{B} \bar{x} )  \,\{e_0;\} }}$ }}

\infax[]
  {\textit{lookup}(m,\bar{C}) = \textit{lookup}_1(m,C_1)}

\infax[]
  {\textit{lookup}_1(m,\text{Object}) = \bullet}

\infrule[]
  {CT(C) = \text{class }C\text{ extends }D\,\{\dots\} \\
   \bar{C} <: \bar{B} \andalso B_0 \ m(\bar{B} \bar{x})\,\{\text{return }e_0;\}  \in \bar{M}
  }
  {\textit{lookup}_1(m,C) = B_0 \ m(\bar{B} \bar{x})\,\{\text{return }e_0;\}}

\infrule[]
  {CT(C) = \text{class }C\text{ extends }D\,\{\dots\} \\
   B_0 \ m(\bar{B} \bar{x})\,\{\text{return }e_0;\}  \not\in \bar{M}
  }
  {\textit{lookup}_1(m,C) = \textit{lookup}_1(m,D)}
  
\medskip

\noindent{\textbf{Method selection: $\boxed{\textit{select}(\bar{m})=m_i}$}}

\smallskip

\infax[]
  {\textit{select}(\bar{m})=m_1}
  
\end{minipage}
}
\caption{Lookup in FJ}
\label{fig:fjlookup}
\end{figure}

\end{document}